\documentclass[letterpaper, conference]{IEEEtran}

\IEEEoverridecommandlockouts
\usepackage{amsmath,amssymb,amsthm,amsfonts,bm}
\usepackage{graphicx}
\usepackage{booktabs}
\usepackage{multirow}
\usepackage{array}
\usepackage{mathtools}
\usepackage{cite}
\usepackage{url}
\usepackage{xcolor}
\usepackage{subcaption}
\usepackage{microtype}
\usepackage{balance}
\graphicspath{{cdc_numerics_three_examples/}}

\usepackage{booktabs}
\usepackage{siunitx}
\newtheorem{theorem}{Theorem}
\newtheorem{proposition}{Proposition}
\newtheorem{remark}{Remark}
\newtheorem{corollary}{Corollary}

\newcommand{\R}{\mathbb{R}}
\newcommand{\1}{\mathbf{1}}
\newcommand{\pbar}{\bar p}
\newcommand{\cbar}{\bar c}
\newcommand{\epsdef}{\epsilon_{\mathrm{def}}}
\newcommand{\epsub}{\epsilon_{\mathrm{ub}}}
\newcommand{\Bset}{\mathcal{B}}
\newcommand{\U}{\mathcal{U}}
\newcommand{\norm}[1]{\lVert #1 \rVert}

\title{Budgeted Robust Intervention Design for Financial Networks with Common Asset Exposures}

\author{%
\IEEEauthorblockN{Giuseppe C. Calafiore} 
\IEEEauthorblockA{{\em Fellow}, IEEE -- Department of Electronics and Telecommunications, Politecnico di Torino, Italy\\
Email: giuseppe.calafiore@polito.it}
}

\begin{document}

\maketitle

\begin{abstract}
In the context of containment of default contagion
in financial networks, we here study a regulator that allocates
pre-shock capital or liquidity buffers across banks connected by
interbank liabilities and common external asset exposures. The
regulator chooses a nonnegative buffer vector under a linear
budget before asset-price shocks realize. Shocks are modeled as
belonging to either an $\ell_{\infty}$ or an $\ell_{1}$ uncertainty set, and the
design objective is either to enlarge the certified no-default/no-
insolvency region or to minimize worst-case clearing losses at
a prescribed stress radius. Four exact synthesis results are
derived. The buffer that maximizes the default resilience margin
is obtained from a linear program and admits a closed-form
minimal-budget certificate for any target margin. The buffer
that maximizes the insolvency resilience margin is computed
by a single linear program. At a fixed radius, minimizing the
worst-case systemic loss is again a linear program under $\ell_{\infty}$
uncertainty and a linear program with one scenario block per
asset under $\ell_{1}$ uncertainty. Crucially, under $\ell_{1}$ uncertainty, exact 
robustness adds only one LP block per asset, ensuring that the computational 
complexity grows linearly with the number of assets. A corollary identifies the exact budget
at which the optimized worst-case loss becomes zero.

Numerical experiments on the 8-bank benchmark of \cite{Calafiore2025}, on a synthetic core-periphery network,  and on a data-backed 107-bank calibration built from the 2025 EBA transparency exercise show large gains over uniform and exposure-proportional allocations. The empirical results also indicate that resilience-maximizing and loss-minimizing interventions nearly coincide under diffuse $\ell_\infty$ shocks, but diverge under concentrated $\ell_1$ shocks.
\end{abstract}

\section{Introduction}
Modern financial systems operate as highly interconnected networks where nodes represent financial institutions and links denote either bilateral obligations, such as interbank loans, or shared vulnerabilities, such as common asset holdings \cite{EisenbergNoe2001, GlassermanYoung2016}. While this interconnectedness facilitates liquidity during normal operation, it also creates pathways for systemic risk \cite{Acemoglu2015}. When external economic shocks reduce the value of a bank's assets, the institution may become unable to fully honor its financial liabilities. This payment shortfall can subsequently propagate to its creditors, potentially triggering a cascade of insolvencies known as {\em default contagion}, \cite{AllenGale2000}. From a systems and control perspective, mitigating this contagion poses a fundamental robust-design challenge: how to optimally allocate protective resources, such as pre-shock capital or liquidity buffers, across the network nodes to guarantee system stability or minimize total losses under worst-case disturbances.

Financial networks may amplify price shocks through two interacting channels. The first is counterparty contagion through interbank liabilities: a payment shortfall at one institution can propagate through the clearing network. The second is market contagion through common asset holdings: a common price move hits several balance sheets simultaneously and may trigger additional defaults through clearing feedback. Both channels are central in the systemic-risk literature \cite{AllenGale2000,EisenbergNoe2001,AllenBabusCarletti2012,Acemoglu2015,GlassermanYoung2016}, and exact LP formulations are available for one-period clearing models \cite{Elsinger2006,Calafiore2024}.

This paper studies a regulator that acts \emph{before} shocks realize. Starting from the bounded-shock stress-testing model of \cite{Calafiore2025}, we add a nonnegative buffer vector $b\in\R^n_+$ subject to $q^\top b\le B$ and ask how this budget should be distributed across banks. The emphasis is therefore on \emph{synthesis}, not on re-deriving robust stress tests: the stress-testing ingredients are inherited from \cite{Calafiore2025}, while the new question is how to turn them into exact pre-shock intervention rules. This viewpoint is complementary to ex-post rescue, threat-index, and default-cost formulations \cite{RogersVeraart2013,Demange2018,Sim2022,AraratMeimanjan2023}.

Relative to [10], the nontrivial step is not simply the
algebraic shift $\bar{c}\mapsto\bar{c}+b$. Once $b$ is chosen before the shock
is observed, the model becomes an outer regulator-design
problem wrapped around an inner robust-clearing problem.
Typically, such formulations involve min-max structures 
that are notoriously difficult to solve globally. However, by carefully applying 
strong duality to the inner worst-case clearing problem and lifting the buffer 
decision to the outer minimization, we demonstrate that the extra design
layer preserves exact tractability: no bilevel optimization,
mixed-integer reformulation, or scenario sampling is needed
to compute the optimal intervention.

Our contributions are fourfold. First, default-margin maximization is shown to be a linear program and yields a closed-form minimal-budget certificate for any target margin. Second, insolvency-margin maximization is again a linear program. Third, worst-case loss minimization at fixed radius is a linear program under $\ell_\infty$ uncertainty. Fourth, the analogous $\ell_1$ problem reduces to a linear program with $m$ scenario blocks, one per asset. An additional corollary identifies the exact budget at which the optimized worst-case loss becomes zero. From a control viewpoint, these formulations separate two objectives: \emph{certification}, namely maximizing the shock radius that can be guaranteed safe, and \emph{performance}, namely minimizing losses once defaults are already admissible.

We offer three numerical examples of increasing realism: an 8-node benchmark from \cite{Calafiore2025}, a synthetic 353-node core-periphery network used to isolate topology effects, and a data-backed 107-bank calibration built from the 2025 EBA transparency exercise \cite{EBA2025}. In the public-data example, interbank row and column totals are extracted from bank-level disclosures and reconciled by iterative proportional fitting/maximum entropy, see \cite{Upper2011,Anand2018}; the common-asset layer uses the five largest sovereign issuer buckets in the sample. Together, these cases show that the LP-based policies remain lightweight beyond the academic benchmark and, more importantly, that the geometry of the uncertainty set materially changes the optimal intervention pattern.

\section{Model and Problem Formulation}
We consider a one-period financial network with $n$ institutions. The nominal liabilities matrix is $\bar P=(\bar p_{ij})\in\R^{n\times n}_+$, where $\bar p_{ij}\geq 0$ is the amount owed by bank $i$ to bank $j$. Let $\pbar=\bar P\1$ denote total liabilities and let $A$ be the relative-liability matrix, defined row-wise by $a_{ij}=\bar p_{ij}/\bar p_i$ when $\bar p_i>0$, and $a_{ii}=1$ otherwise \cite{EisenbergNoe2001,Calafiore2025}.

Banks are exposed to $m$ external assets. Let $S\in\R^{n\times m}$ be the portfolio matrix, where $s_{ij}$ is the position of bank $i$ in asset $j$; long and short positions are both allowed. The nominal external net inflow vector is $\cbar\in\R^n$, the price perturbation is $\delta\in\R^m$, and the regulator injects the pre-shock buffer $b\in\R^n_+$. The realized net inflow is
\begin{equation}
    c=\cbar+b+S\delta.
    \label{eq:realized_c}
\end{equation}
The admissible intervention set is
\begin{equation}
    \Bset(B)\doteq\{b\in\R^n_+:\ q^\top b\le B\},
    \label{eq:budget_set}
\end{equation}
where $q\in\R^n_{++}$ collects per-unit intervention costs.

\subsection{Clearing and systemic loss}
For a realized net inflow $c$, the maximal clearing vector is obtained from the LP \cite{Calafiore2024,Calafiore2025}
\begin{equation}
\begin{aligned}
    \eta^\star(c)=\min_{p\in\R^n}\;&\1^\top(\pbar-p)\\
    \text{s.t. }&0\le p\le \pbar,\\
    &c+A^\top p\ge p.
\end{aligned}
\label{eq:clearing_lp}
\end{equation}
When feasible, \eqref{eq:clearing_lp} returns the maximal clearing vector and the associated systemic loss $\eta^\star(c)$; when infeasible we set $\eta^\star(c)=+\infty$. Throughout the paper we assume nominal no-default operation,
\begin{equation}
    \cbar+A^\top\pbar>\pbar,
    \label{eq:nominal_no_default}
\end{equation}
which implies the strictly positive nominal net-worth margin
\begin{equation}
    r\doteq\cbar+(A^\top-I)\pbar>0.
    \label{eq:r_def}
\end{equation}

We consider the uncertainty sets
\begin{equation}
\begin{aligned}
    \U_\infty(\epsilon)&=\{\delta\in\R^m:\ \norm{\delta}_\infty\le\epsilon\},\\
    \U_1(\epsilon)&=\{\delta\in\R^m:\ \norm{\delta}_1\le\epsilon\}.
\end{aligned}
\label{eq:uncertainty_sets}
\end{equation}
These two uncertainty sets encode different stress geometries. The $\ell_\infty$ set models diffuse simultaneous price moves across many assets, whereas the $\ell_1$ set captures concentrated stress with fixed total magnitude. Correspondingly, the bank-level exposure score becomes $\alpha_i=\norm{\sigma_i}_1$ in the $\ell_\infty$ case and $\alpha_i=\norm{\sigma_i}_\infty$ in the $\ell_1$ case (here and later,  $\sigma_i^\top$ denotes the $i$th row of $S$). This distinction will explain why diffuse and concentrated shocks induce markedly different buffer patterns.
For fixed $b$, define the worst-case loss
\begin{equation}
    J_p(\epsilon;b)\doteq
    \max_{\delta\in\U_p(\epsilon)}\eta^\star(\cbar+b+S\delta),
    \qquad p\in\{1,\infty\}.
    \label{eq:worst_case_loss_fixed_b}
\end{equation}
Our goal is to choose $b\in\Bset(B)$ to shape both resilience margins and worst-case losses.
We therefore separate two robust certificates: the default margin, which requires full clearing for all admissible shocks, and the insolvency margin of \cite{Calafiore2025}, which only requires feasibility of the clearing LP. The first is stricter, while the second still rules out infeasibility  of the one-period clearing problem.

\section{Resilience-Optimal Buffer Design}
\subsection{Default resilience margin}
Let $\sigma_i^\top$ be row $i$ of $S$, and let $\alpha_i\doteq\norm{\sigma_i}_*$, where $\norm{\cdot}_*$ is the dual norm associated with the chosen uncertainty set. When $\alpha_i=0$, we adopt the standard convention $(r_i+b_i)/\alpha_i=+\infty$.

\begin{proposition}
For any fixed $b\in\R^n_+$,
\begin{equation}
    \epsdef(b)=\min_{i=1,\dots,n}\frac{r_i+b_i}{\alpha_i}.
    \label{eq:epsdef_fixedb}
\end{equation}
Hence $\epsdef(b)$ is the largest shock radius for which every $\delta\in\U_p(\epsilon)$ yields full clearing $p=\pbar$.
\end{proposition}

The design problem is
\begin{equation}
    \max_{b\in\Bset(B)}\ \epsdef(b).
    \label{eq:max_default_margin}
\end{equation}

\begin{theorem}
Let $r$ be given by \eqref{eq:r_def} and $\alpha_i=\norm{\sigma_i}_*$. The optimal value of \eqref{eq:max_default_margin} is the solution of the LP
\begin{equation}
\begin{aligned}
    \max_{\epsilon,\,b}\;&\epsilon\\
    \text{s.t. }&b_i-\alpha_i\epsilon\ge-r_i,\quad i=1,\dots,n,\\
    &b\ge0,\quad q^\top b\le B.
\end{aligned}
\label{eq:default_margin_lp}
\end{equation}
Moreover, for every target radius $\epsilon\ge0$, the minimal budget that guarantees $\epsdef(b)\ge\epsilon$ is
\begin{equation}
    B_{\mathrm{def}}(\epsilon)=\sum_{i=1}^n q_i\,[\alpha_i\epsilon-r_i]_+,
    \label{eq:min_budget_formula}
\end{equation}
with corresponding minimal buffer $b_i^\epsilon=[\alpha_i\epsilon-r_i]_+$.
\end{theorem}

\begin{proof}
From \eqref{eq:epsdef_fixedb}, the condition $\epsdef(b)\ge\epsilon$ is equivalent to $r_i+b_i\ge\alpha_i\epsilon$ for all $i$, which gives \eqref{eq:default_margin_lp}. For fixed $\epsilon$, the least expensive feasible choice is obtained componentwise by selecting the smallest admissible $b_i$, namely $b_i^\epsilon=[\alpha_i\epsilon-r_i]_+$. Summing the costs gives \eqref{eq:min_budget_formula}.
\end{proof}

\begin{remark}
The certificate \eqref{eq:min_budget_formula} has a simple top-up interpretation. To guarantee a target radius $\epsilon$, bank $i$ receives support only when its nominal margin $r_i$ falls short of the required stress buffer $\alpha_i\epsilon$. Thus the optimal design is not merely exposure-proportional: it is exposure-adjusted and balance-sheet-aware, and it concentrates spending on the smallest margin-to-exposure ratios.
\end{remark}

\begin{corollary}
The optimal resilience value $\epsdef^\star(B)=\max_{b\in\Bset(B)}\epsdef(b)$ is nondecreasing, concave, and piecewise-affine in $B$.
\end{corollary}

\begin{proof}
The map $H(\epsilon)=\sum_i q_i[\alpha_i\epsilon-r_i]_+$ is nondecreasing, convex, and piecewise-affine. Since $\epsdef^\star(B)=\max\{\epsilon:H(\epsilon)\le B\}$, it is the inverse of $H$ on its increasing range.
\end{proof}

\subsection{Insolvency resilience margin}
Applying Theorem~3 of \cite{Calafiore2025} to the shifted net inflow $\cbar+b$ yields the design LP below.

\begin{theorem}
Let $s\in\R^n_+$ be defined by $s_i=\norm{\sigma_i}_*$. The largest insolvency resilience margin achievable with budget $B$ is the optimal value of
\begin{equation}
\begin{aligned}
    \max_{\epsilon,\,p,\,b}\;&\epsilon\\
    \text{s.t. }&0\le p\le\pbar,\\
    &\cbar+b-\epsilon s+A^\top p\ge p,\\
    &b\ge0,\quad q^\top b\le B.
\end{aligned}
\label{eq:epsub_design_lp}
\end{equation}
\end{theorem}

\begin{proof}
For fixed $b$, Theorem~3 in \cite{Calafiore2025} states that the insolvency resilience margin is the optimal value of the feasibility LP obtained by replacing $\cbar$ with $\cbar+b$. Optimizing over all $b\in\Bset(B)$ gives \eqref{eq:epsub_design_lp}.
\end{proof}

\begin{remark}
For every admissible buffer $b$, one has $\epsdef(b)\le \epsub(b)$, since full clearing is a feasible point of \eqref{eq:clearing_lp}. The gap between the two margins therefore quantifies a regime in which the system can remain feasible even though some defaults may already occur.
\end{remark}

\section{Loss-Optimal Buffer Design}
We now minimize the worst-case systemic loss at a prescribed radius $\epsilon$. The key point is that the problem remains exactly tractable after adding the buffer decision.

\subsection{$\ell_\infty$-bounded perturbations}
For $\U_\infty(\epsilon)$ define the bank-level stress vector
\begin{equation}
    s^\infty\doteq|S|\1,
    \qquad
    s_i^\infty=\sum_{j=1}^m|s_{ij}|=\norm{\sigma_i}_1.
    \label{eq:sinf_def}
\end{equation}

\begin{theorem}
Let
\begin{equation}
    J_\infty(\epsilon,B)\doteq\min_{b\in\Bset(B)}J_\infty(\epsilon;b).
    \label{eq:Jinf_budget_def}
\end{equation}
Then $J_\infty(\epsilon,B)$ is the optimal value of the LP
\begin{equation}
\begin{aligned}
    \min_{p,\,b}\;&\1^\top(\pbar-p)\\
    \text{s.t. }&0\le p\le\pbar,\\
    &\cbar+b-\epsilon s^\infty+A^\top p\ge p,\\
    &b\ge0,\quad q^\top b\le B.
\end{aligned}
\label{eq:Jinf_design_lp}
\end{equation}
If \eqref{eq:Jinf_design_lp} is infeasible, then $J_\infty(\epsilon,B)=+\infty$.
\end{theorem}

\begin{proof}
For fixed $b$, Eq.~(34) of \cite{Calafiore2025} writes the worst-case loss as
\begin{equation}
\begin{aligned}
J_\infty(\epsilon;b)=\max_{\beta,\lambda\ge0}\;&
(\1-\beta)^\top\pbar-(\cbar+b)^\top\lambda
+\epsilon(s^\infty)^\top\lambda\\
\text{s.t. }&\beta-\1+(I-A)\lambda\ge0.
\end{aligned}
\label{eq:dual_inf_fixedb}
\end{equation}
Dualizing \eqref{eq:dual_inf_fixedb} with dual variable $p\ge0$ gives \eqref{eq:Jinf_design_lp} for fixed $b$; minimizing over $b\in\Bset(B)$ gives the stated LP. Infeasibility means that some admissible shock in $\U_\infty(\epsilon)$ causes insolvency toward the external sector for every admissible buffer.
\end{proof}

\subsection{$\ell_1$-bounded perturbations}
For $k=1,\dots,m$, let $\zeta_k^\top$ denote row $k$ of $S^\top$, so $|\zeta_k|\in\R^n_+$ is the cross-sectional exposure profile to asset $k$.

\begin{theorem}
Let
\begin{equation}
    J_1(\epsilon,B)\doteq\min_{b\in\Bset(B)}J_1(\epsilon;b).
    \label{eq:J1_budget_def}
\end{equation}
Then $J_1(\epsilon,B)$ is the optimal value of
\begin{equation}
\begin{aligned}
    \min_{t,\,b,\,p^{(1)},\dots,p^{(m)}}\;&t\\
    \text{s.t. }&t\ge\1^\top(\pbar-p^{(k)}),\quad k=1,\dots,m,\\
    &0\le p^{(k)}\le\pbar,\quad k=1,\dots,m,\\
    &\cbar+b-\epsilon|\zeta_k|+A^\top p^{(k)}\ge p^{(k)},\ k=1,\dots,m,\\
    &b\ge0,\quad q^\top b\le B.
\end{aligned}
\label{eq:J1_design_lp}
\end{equation}
If \eqref{eq:J1_design_lp} is infeasible, then $J_1(\epsilon,B)=+\infty$.
\end{theorem}

\begin{proof}
For fixed $b$, Eq.~(31) of \cite{Calafiore2025} yields
\begin{equation}
J_1(\epsilon;b)=\max_{k=1,\dots,m}\ \max_{\beta,\lambda\ge0}
\Big[(\1-\beta)^\top\pbar-(\cbar+b)^\top\lambda
+\epsilon|\zeta_k|^\top\lambda\Big]
\label{eq:J1_dual_compact}
\end{equation}
subject to $\beta-\1+(I-A)\lambda\ge0$. For each fixed $k$, dualizing the inner LP gives
\begin{equation}
\begin{aligned}
\phi_k(\epsilon;b)=\min_{p^{(k)}}\;&\1^\top(\pbar-p^{(k)})\\
\text{s.t. }&0\le p^{(k)}\le\pbar,\\
&\cbar+b-\epsilon|\zeta_k|+A^\top p^{(k)}\ge p^{(k)}.
\end{aligned}
\label{eq:phi_k_lp}
\end{equation}
Hence $J_1(\epsilon;b)=\max_k\phi_k(\epsilon;b)$. Introducing the epigraph variable $t$ and minimizing over $b\in\Bset(B)$ gives \eqref{eq:J1_design_lp}; infeasibility means that at least one asset-scenario block is infeasible for every admissible buffer.
\end{proof}

\begin{remark}
The tractability gap between the two uncertainty models is mild. Under $\ell_\infty$ uncertainty, robustness collapses to the single stress vector $s^\infty$ in \eqref{eq:sinf_def}. Under $\ell_1$ uncertainty, exact robustness adds only one LP block per asset because the worst sign is already absorbed by $|\zeta_k|$. Hence the complexity growth is linear in the asset dimension $m$.
\end{remark}

\begin{corollary}
For fixed $\epsilon$, the value functions $J_\infty(\epsilon,B)$ and $J_1(\epsilon,B)$ are nonincreasing, convex, and piecewise-affine in $B$ on their feasible domain.
\end{corollary}

\begin{proof}
Both \eqref{eq:Jinf_design_lp} and \eqref{eq:J1_design_lp} are parametric LPs with scalar parameter $B$ only in the right-hand side of the budget inequality. Standard parametric LP theory yields the claim \cite{Boyd2004,Bertsimas2011}.
\end{proof}

\begin{corollary}
For either uncertainty set $\U_p(\epsilon)$, the optimized worst-case loss satisfies
\begin{equation}
    J_p(\epsilon,B)=0\iff B\ge B_{\mathrm{def}}(\epsilon),
    \label{eq:zero_loss_budget}
\end{equation}
where $B_{\mathrm{def}}(\epsilon)$ is given by \eqref{eq:min_budget_formula} with the dual norm corresponding to $p\in\{1,\infty\}$.
\label{cor:zero_loss_budget}
\end{corollary}

\begin{proof}
$J_p(\epsilon,B)=0$ holds if and only if some $b\in\Bset(B)$ makes full clearing $p=\pbar$ feasible for every $\delta\in\U_p(\epsilon)$. This is equivalent to $\epsdef(b)\ge\epsilon$, because zero systemic loss means no defaults in any admissible scenario. Theorem~1 then gives the exact minimal budget.
\end{proof}

Equation \eqref{eq:zero_loss_budget} is operationally useful because it identifies the exact budget below which any loss-optimal design must accept positive residual contagion. Numerical results are eventually summarized in Table~\ref{tab:key_numbers}.\footnote{In Table~\ref{tab:key_numbers}, $\epsilon_p^\star(B)$ denotes the optimal default resilience margin at budget $B$; in these rows the `Margin' column is void 
because this value is the same  reported in the margin-optimal allocation column ‘Opt.’
In the 
$J_p(\epsilon,B)$-type rows,
the column ‘Margin’ gives the worst-case loss obtained, at the same stress radius 
$\epsilon$, by the allocation that maximizes the default resilience margin under the same budget $B$.}

\section{Numerical Examples}
All computations were performed with the HiGHS backend of \texttt{scipy.optimize.linprog}. For comparison we use two simple baselines. The uniform baseline splits intervention spending equally across banks, i.e., $q_i b_i=B/n$. The exposure-proportional baseline allocates spending proportionally to each bank's common-asset stress score, i.e., $q_i b_i=B\alpha_i/\sum_j\alpha_j$, where $\alpha_i$ is the score corresponding to the uncertainty set of the panel under consideration. In the loss plots, ``margin-optimal'' denotes the buffer obtained from \eqref{eq:default_margin_lp} at the same budget. We report three examples of increasing realism: an 8-node benchmark, a synthetic 353-node core-periphery network, and a public-data 107-bank EBA calibration.

\begin{figure*}[t]
    \centering
    \begin{subfigure}[t]{0.48\textwidth}
        \centering
        \includegraphics[width=\linewidth]{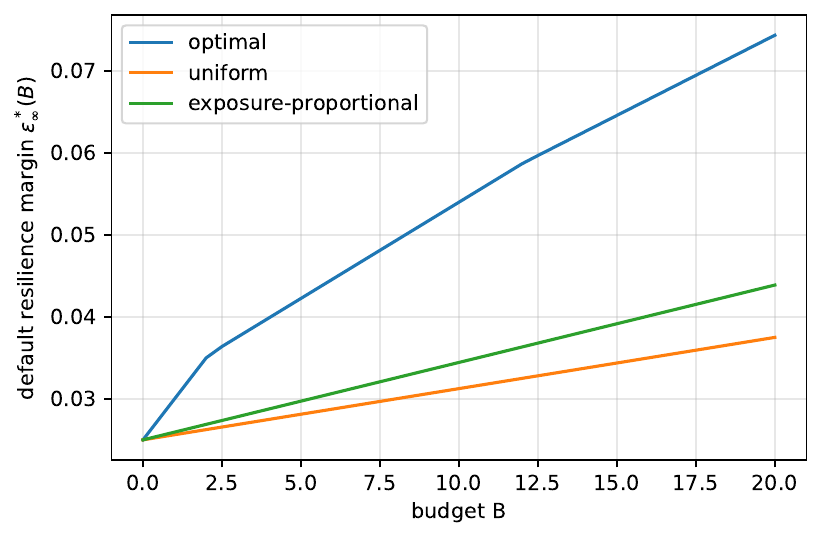}
        \caption{$\ell_\infty$ default margin.}
    \end{subfigure}\hfill
    \begin{subfigure}[t]{0.48\textwidth}
        \centering
        \includegraphics[width=\linewidth]{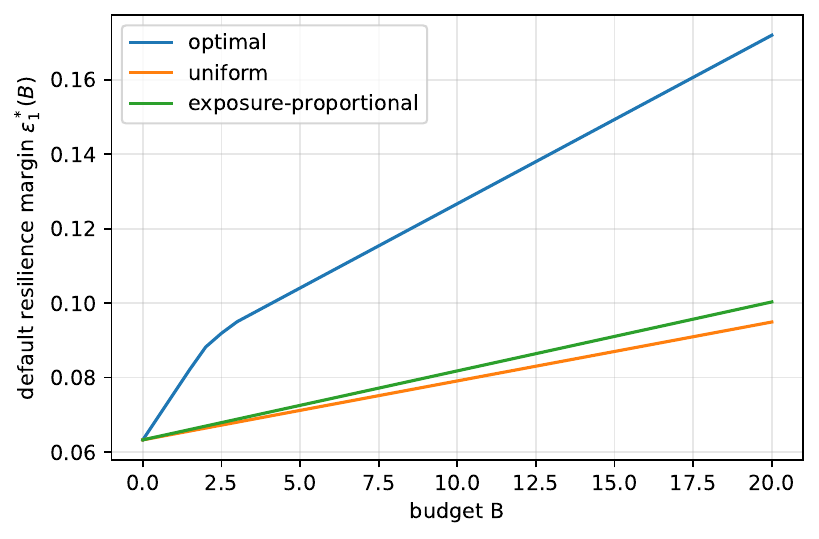}
        \caption{$\ell_1$ default margin.}
    \end{subfigure}

    \vspace{0.6em}
    \begin{subfigure}[t]{0.48\textwidth}
        \centering
        \includegraphics[width=\linewidth]{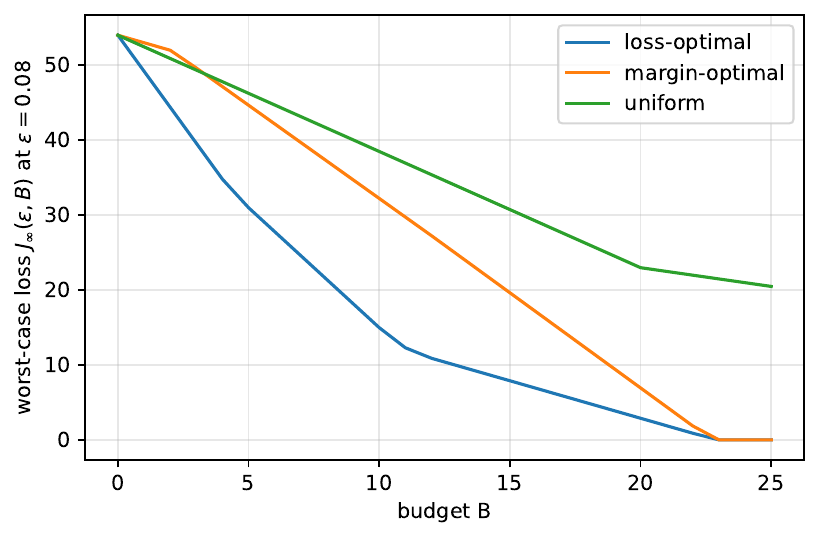}
        \caption{$J_\infty(0.08,B)$.}
    \end{subfigure}\hfill
    \begin{subfigure}[t]{0.48\textwidth}
        \centering
        \includegraphics[width=\linewidth]{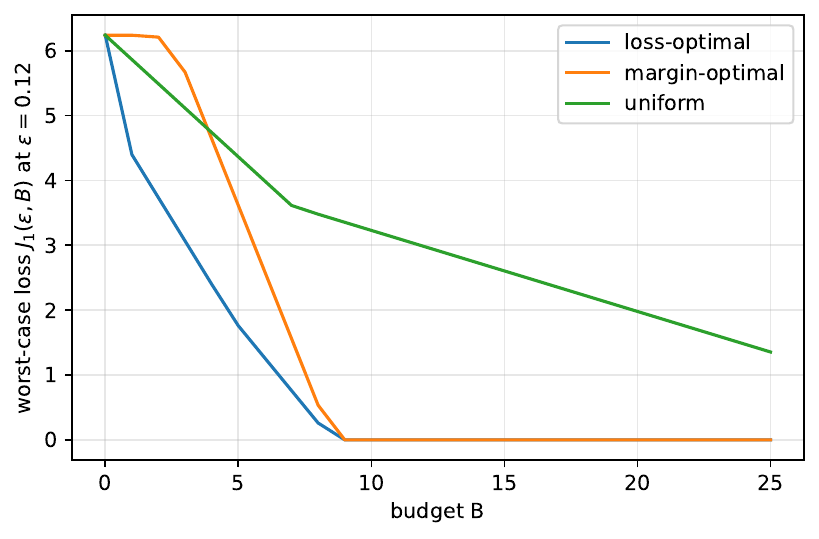}
        \caption{$J_1(0.12,B)$.}
    \end{subfigure}
    \caption{Eight-node benchmark. Optimal intervention dominates simple heuristics for both resilience maximization and loss minimization.}
    \label{fig:ex1_main}
\end{figure*}

\subsection{Eight-node benchmark}
We first reuse the 8-node, 4-asset benchmark in \cite{Calafiore2025}. With no intervention,
\begin{equation*}
\epsdef(0)=0.0250\ (\ell_\infty),\qquad \epsdef(0)=0.0633\ (\ell_1),
\end{equation*}
while the insolvency margins are $\epsub(0)=0.1875$ in the $\ell_\infty$ case and $\epsub(0)=0.4327$ in the $\ell_1$ case. Unit intervention costs $q_i\equiv1$ are used.

Fig.~\ref{fig:ex1_main}(a)-(b) shows the budget-response of the default resilience margin. At budget $B=10$, the optimal design lifts the no-default radius from $0.0250$ to $0.0540$ under $\ell_\infty$ shocks, compared with $0.0313$ for a uniform allocation and $0.0344$ for an exposure-proportional one. Under $\ell_1$ shocks the same budget raises the margin from $0.0633$ to $0.1267$, while the two heuristics reach only $0.0791$ and $0.0818$. This behavior matches Theorem~1: the optimal policy spends budget to raise the smallest normalized balance-sheet margins.

Fig.~\ref{fig:ex1_main}(c)-(d) shows that resilience maximization and loss minimization are different design goals. For $\epsilon=0.08$ and $B=10$, the loss-optimal policy reduces $J_\infty(\epsilon,B)$ from $53.96$ to $14.96$, whereas the margin-optimal policy still yields $32.19$ and the uniform allocation $38.46$. In the $\ell_1$ case, for $\epsilon=0.12$ and $B=5$, the loss-optimal design gives $J_1=1.76$ versus $3.62$ for the margin-optimal allocation and $4.37$ for the uniform one. Corollary~\ref{cor:zero_loss_budget} predicts exact zero-loss budgets $B_{\mathrm{def}}(0.08)=22.88$ for $\ell_\infty$ and $B_{\mathrm{def}}(0.12)=8.52$ for $\ell_1$, which coincide with the points where the optimized curves hit zero.

\begin{table}[t]
\centering
\caption{Optimal resilience margins and worst-case losses under different
buffer allocation rules.\label{tab:key_numbers}
}
\resizebox{\columnwidth}{!}{%
\begin{tabular}{ll
                S[table-format=2.4]
                c
                S[table-format=2.4]
                S[table-format=2.4]}
\toprule
Case & Metric & {Opt.} & {Margin} & {Uniform} & {Exp.-prop.} \\
\midrule
8-node
& $\varepsilon_\infty^*(10)$ & 0.0540 & -- & 0.0313 & 0.0344 \\
& $\varepsilon_1^*(10)$      & 0.1267 & -- & 0.0791 & 0.0818 \\
& $J_\infty(0.08,10)$        & 14.96  & 32.19 & 38.46 & 31.56 \\
& $J_1(0.12,5)$              & 1.76   & 3.62  & 4.37  & 4.14 \\
\midrule
CP-353
& $\varepsilon_\infty^*(25)$ & 0.0345 & -- & 0.0216 & 0.0219 \\
& $J_\infty(0.04,25)$        & 22.25  & 38.45 & 72.03 & 71.74 \\
\midrule
EBA-107
& $\varepsilon_\infty^*(1)$  & 0.1088 & -- & 0.0818 & 0.0817 \\
& $\varepsilon_1^*(5)$       & 0.1385 & -- & 0.1147 & 0.1162 \\
& $J_\infty(0.108,0.5)$      & 0.3148 & 0.3148 & 0.8199 & 0.8018 \\
& $J_1(0.135,3)$             & 0.9594 & 1.0427 & 3.9314 & 3.7552 \\
\bottomrule
\end{tabular}
}
\end{table}

\begin{figure*}[t]
    \centering
    \begin{subfigure}[t]{0.48\textwidth}
        \centering
        \includegraphics[width=\linewidth]{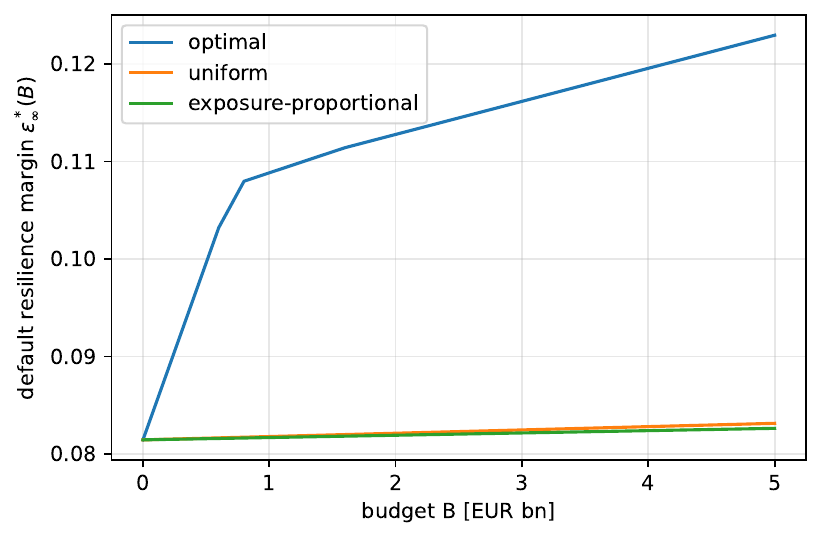}
        \caption{$\ell_\infty$ default margin.}
    \end{subfigure}\hfill
    \begin{subfigure}[t]{0.48\textwidth}
        \centering
        \includegraphics[width=\linewidth]{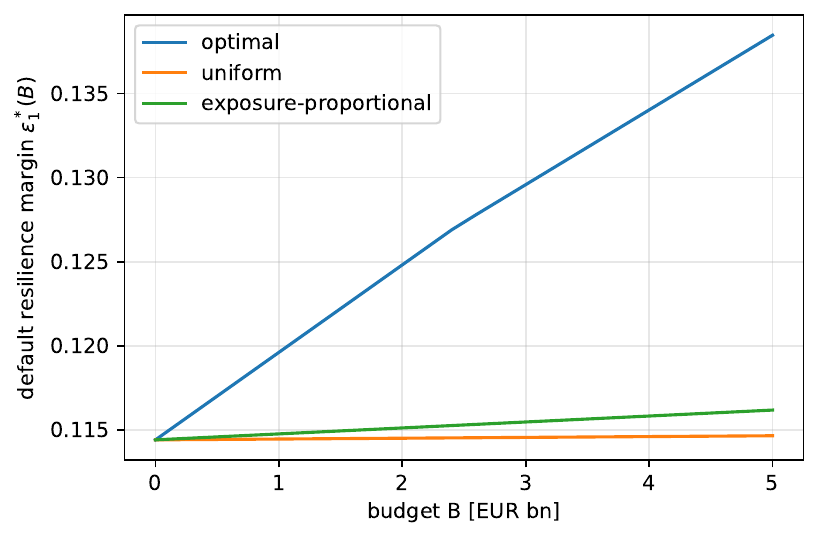}
        \caption{$\ell_1$ default margin.}
    \end{subfigure}

    \vspace{0.6em}
    \begin{subfigure}[t]{0.48\textwidth}
        \centering
        \includegraphics[width=\linewidth]{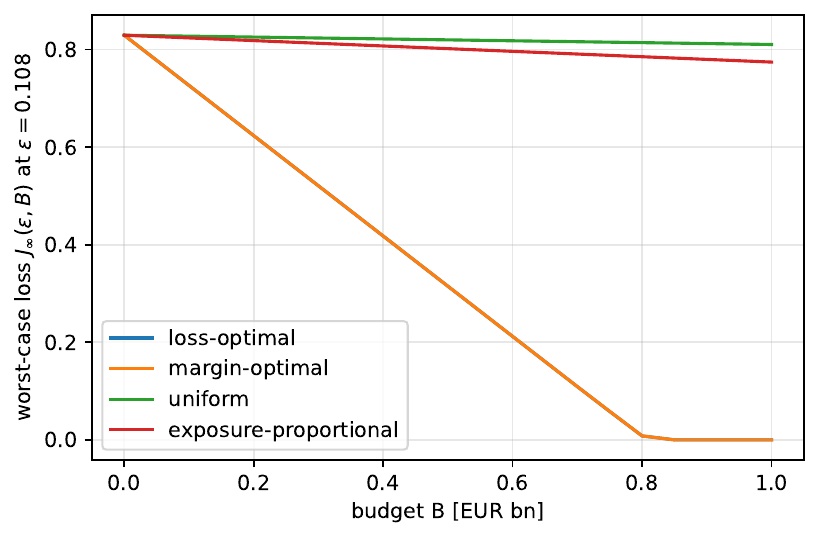}
        \caption{$J_\infty(0.108,B)$.}
    \end{subfigure}\hfill
    \begin{subfigure}[t]{0.48\textwidth}
        \centering
        \includegraphics[width=\linewidth]{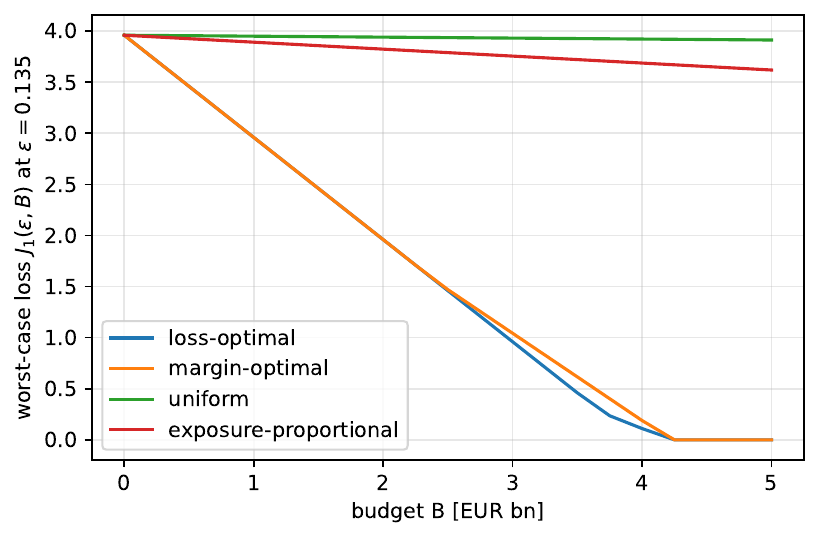}
        \caption{$J_1(0.135,B)$.}
    \end{subfigure}
    \caption{Data-backed 107-bank EBA calibration. Diffuse $\ell_\infty$ shocks make margin-optimal and loss-optimal designs nearly coincide, whereas concentrated $\ell_1$ shocks separate the two objectives.}
    \label{fig:eba_main}
\end{figure*}

\subsection{Synthetic 353-node core-periphery network}
We next discuss a synthetic 353-node, 5-asset core-periphery network. The instance has an 18-node dense core and a sparse periphery, with random liabilities and common-asset holdings generated with seed 42. Without intervention,
\begin{equation*}
\epsdef(0)=0.0200\ (\ell_\infty),\qquad \epsdef(0)=0.0208\ (\ell_1),
\end{equation*}
while the insolvency margins are $\epsub(0)=0.0493$ in the $\ell_\infty$ case and $\epsub(0)=0.0670$ in the $\ell_1$ case. At budget $B=25$, the optimal $\ell_\infty$ design raises the default margin to $0.0345$, versus $0.0216$ for uniform allocation and $0.0219$ for an exposure-proportional rule; see Table~\ref{tab:key_numbers}.

At the loss-design radius $\epsilon=0.04$, Fig.~\ref{fig:cp_loss} shows that the loss-optimal policy reduces the worst-case loss from $97.85$ to $22.25$, whereas the margin-optimal allocation still yields $38.45$ and the uniform rule $72.03$. Because the instance is synthetic, its role is not empirical calibration but isolation of a topology effect: most of the gain comes from concentrating support on a small subset of core institutions rather than spreading it across the periphery.

\subsection{Data-backed 107-bank EBA calibration}
Our third example uses the 2025 EBA EU-wide transparency exercise, reference date 30 June 2025, which reports bank-level supervisory disclosures for 119 banks across 25 EU/EEA countries \cite{EBA2025}. We retain the 107 institutions with positive interbank-asset proxy, positive interbank liabilities, and positive equity. Interbank asset totals are proxied by exposure value to institutions from the credit-risk templates, interbank liabilities by deposits from credit institutions from the liabilities template, and equity by total equity. 
Because public row and column
totals do not balance exactly on the filtered sample, both
sides are multiplicatively reconciled to a common aggregate
and a bilateral liabilities matrix is reconstructed by iterative
proportional fitting [4], [20]. This maximum-entropy approach 
is widely adopted in empirical network stress-testing because it generates 
the most dispersed, least biased network topology consistent with the 
observed marginals, thus preventing the spurious concentration of contagion risk.
The common-asset layer is formed
by reported sovereign holdings in the five largest issuer buckets
at that date: FR, IT, ES, US, and DE. All monetary quantities
are expressed in EUR bn.

Without intervention, we have
\begin{equation*}
\epsdef(0)=0.0814\ (\ell_\infty),\qquad \epsdef(0)=0.1144\ (\ell_1),
\end{equation*}
while the insolvency margins are $\epsub(0)=0.1099$ in the $\ell_\infty$ case and $\epsub(0)=0.1372$ in the $\ell_1$ case. Fig.~\ref{fig:eba_main}(a)-(b) shows that targeted intervention is dramatically more effective than simple heuristics. At budget $B=1$, the optimal $\ell_\infty$ design raises the default margin to $0.1088$, versus $0.0818$ for uniform allocation and $0.0817$ for exposure-proportional allocation. At budget $B=5$, the optimal $\ell_1$ design reaches $0.1385$, while the two heuristics remain at $0.1147$ and $0.1162$.

The loss-design results reveal a sharp dependence on shock geometry. Under diffuse $\ell_\infty$ stress, Fig.~\ref{fig:eba_main}(c) shows that the loss-optimal and margin-optimal curves coincide on this instance and hit zero exactly at the certified threshold $B_{\mathrm{def}}(0.108)=0.8080$, whereas uniform and exposure-proportional allocations remain far from zero. In contrast, under concentrated $\ell_1$ stress, Fig.~\ref{fig:eba_main}(d) shows a clear separation between the two objectives: at $\epsilon=0.135$ and $B=3$, the loss-optimal design yields $J_1=0.9594$, compared with $1.0427$ for the margin-optimal design, $3.9314$ for uniform allocation, and $3.7552$ for exposure-proportional allocation. The corresponding zero-loss threshold is $B_{\mathrm{def}}(0.135)=4.2209$. The optimal $\ell_1$ allocations are highly concentrated: at $B=3$, loss minimization places the entire budget on one bank, whereas resilience maximization splits it across two banks, and the exposure-proportional rule spreads support across 94 institutions. Thus the public-data calibration reaches the same qualitative conclusion as the 8-node benchmark example, but with a stronger empirical angle: diffuse shocks align certification and performance, whereas concentrated shocks separate them.

\begin{figure}[t]
    \centering
    \includegraphics[width=\columnwidth]{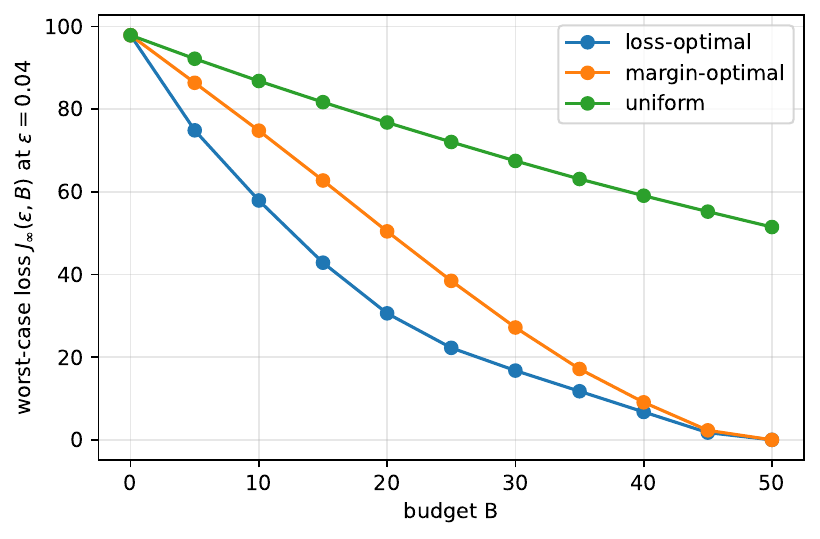}
    \caption{Synthetic 353-node core-periphery network. At $\epsilon=0.04$, topology-aware loss-optimal design markedly outperforms margin-optimal and uniform allocations.}
    \label{fig:cp_loss}
\end{figure}

\section{Conclusions}
We introduced a budgeted pre-shock intervention layer in a robust financial contagion model with common asset exposures. The resulting regulator problems remain exact and computationally simple: default-margin maximization admits both an LP formulation and a closed-form minimal-budget certificate, insolvency-margin maximization is an LP, and loss minimization is an LP under both $\ell_\infty$ and $\ell_1$ uncertainty. The additional zero-loss budget corollary makes explicit when robust certification and robust loss minimization coincide.

The numerical results show three practical lessons. First, targeted interventions are far more effective than uniform ones, even under small budgets. Second, the best allocation depends strongly on both the design objective and the shock geometry: maximizing the shock radius that can be certified as safe is not the same as minimizing the loss once defaults are already admissible, especially under concentrated $\ell_1$ shocks. Third, the method scales beyond the academic benchmark. In the 353-node synthetic network, the $\ell_\infty$ loss-design LP has 706 decision variables; in the 107-bank EBA calibration, the corresponding $\ell_\infty$ LP has 214 variables and the $\ell_1$ formulation has 643 because of the five scenario blocks. These sizes remained easily manageable with HiGHS through \texttt{scipy.optimize.linprog}, making the approach practical for repeated stress-testing, budget scans, and outer design loops. Dynamic intervention design \cite{Calafiore2023}, debt-forgiveness and fire-sale mechanisms \cite{RogersVeraart2013,FeinsteinHallaj2023,AraratMeimanjan2023}, and topology co-design \cite{Hu2024} appear as natural next steps for research.

Two caveats deserve emphasis. First, the model is one-period and the intervention is static: buffers are chosen once, before shocks realize, and are not adapted over time. Second, asset-price effects enter through exogenous uncertainty sets rather than endogenous fire-sale dynamics, regulatory constraints, or strategic responses. The present contribution should therefore be viewed as an exact inner design layer that can be embedded in richer stress-testing architectures.

\balance

\end{document}